\documentclass[conference]{IEEEtran}
\IEEEoverridecommandlockouts
\usepackage{tikz,xcolor}
\usepackage[hidelinks, draft]{hyperref}
\usetikzlibrary{graphs}
\usepackage{url}  

\usepackage{balance}
\usepackage{stfloats}

\definecolor{lime}{HTML}{A6CE39}
\DeclareRobustCommand{\orcidicon}{%
	\begin{tikzpicture}
	\draw[lime, fill=lime] (0,0) 
	circle [radius=0.16] 
	node[white] {{\fontfamily{qag}\selectfont \tiny ID}};
	\draw[white, fill=white] (-0.0625,0.095) 
	circle [radius=0.007];
	\end{tikzpicture}
	\hspace{-2mm}
}

\foreach \x in {A, ..., Z}{%
	\expandafter\xdef\csname orcid\x\endcsname{\noexpand\href{https://orcid.org/\csname orcidauthor\x\endcsname}{\noexpand\orcidicon}}
}
\usepackage{cite}
\usepackage{amsmath,amssymb,amsfonts,amsthm}
\usepackage{algorithmic}
\usepackage{graphicx}
\usepackage{textcomp}
\usepackage{xcolor}
\usepackage{comment}
\usepackage{float}
\usepackage{wrapfig}
\usepackage{caption} 
\def\BibTeX{{\rm B\kern-.05em{\sc i\kern-.025em b}\kern-.08em
    T\kern-.1667em\lower.7ex\hbox{E}\kern-.125emX}}

\newtheorem{Definition}{Definition}

\newtheorem{Example}{Example}
\newtheorem{Theorem}{Theorem}

\newcommand{\routid}[1]{\mathbf{R}_{\mathrm{Id}}(#1)}
\newcommand{\routuni}[1]{\mathbf{R}_{\mathrm{Uni}}(#1)}
\newcommand{\routcut}[1]{\mathbf{R}_{\mathrm{Cut}}(#1)}
\newcommand{\hypid}[1]
{\mathcal{H}_{\mathrm{Id}}(#1)}
\newcommand{\hypuni}[1]{\mathcal{H}_{\mathrm{Uni}}(#1)}
\newcommand{\hypcut}[1]{\mathcal{H}_{\mathrm{Cut}}(#1)}

\newcommand{\weightdel}{w_{\mathrm{Del}}}
\newcommand{\weightcost}{w_{\mathrm{Cost}}}
\newcommand{\weightcap}{w_{\mathrm{Cap}}}
\newcommand{\weightupr}{w_{\mathrm{UPr}}}
\newcommand{\weightfpr}{w_{\mathrm{FPr}}}

\newcommand{\rhodel}{\rho_{\mathrm{Del}}}
\newcommand{\rhocost}{\rho_{\mathrm{Cost}}}
\newcommand{\rhocap}{\rho_{\mathrm{Cap}}}
\newcommand{\rhoupr}{\rho_{\mathrm{UPr}}}
\newcommand{\rhofpr}{\rho_{\mathrm{FPr}}}
\newcommand{\rhoapr}{1 - \rho_{\mathrm{UPr}}}
\newcommand{\rhospr}{1 - \rho_{\mathrm{FPr}}}

\usepackage{subcaption}
\usepackage{booktabs}

\begin{document}

\title{A Formal Model for Path Set Attribute Calculation in Network Systems
}

% add for author formatting
\makeatletter
\newcommand{\linebreakand}{%
  \end{@IEEEauthorhalign}
  \hfill\mbox{}\par
  \mbox{}\hfill\begin{@IEEEauthorhalign}
}
\makeatother
% end of author formatting
\author{
\IEEEauthorblockN{Giovanni Fiaschi\orcidA{}}
\IEEEauthorblockA{\textit{Radio New Concept and Algorithms} \\
\textit{Ericsson AB}\\
Stockholm, Sweden\\
Giovanni.Fiaschi@ericsson.com}
\and
\IEEEauthorblockN{Carlo Vitucci\orcidB{}}
\IEEEauthorblockA{\textit{Technology Management} \\
\textit{Ericsson AB}\\
Stockholm, Sweden \\
Carlo.Vitucci@ericsson.com}
\and
\IEEEauthorblockN{Thomas Westerbäck\orcidC{}}
\IEEEauthorblockA{\textit{Division of Mathematics and Physics} \\
\textit{Mälardalen University}\\
Västerås, Sweden \\
thomas.westerback@mdu.se}
\linebreakand
\IEEEauthorblockN{Daniel Sundmark\orcidD{}}
\IEEEauthorblockA{\textit{Computer Science and Software Enigineering} \\
\textit{Mälardalen University}\\
Västerås, Sweden \\
daniel.sundmark@mdu.se}
\and
\IEEEauthorblockN{Thomas Nolte\orcidE{}}
\IEEEauthorblockA{\textit{Division of Networked and Embedded Systems} \\
\textit{Mälardalen University}\\
Västerås, Sweden \\
thomas.nolte@mdu.se}
}

\maketitle

\begin{abstract}
In graph theory and its practical networking applications, e.g., telecommunications and transportation, the problem of finding paths has particular importance. Selecting paths requires giving scores to the alternative solutions to drive a choice. 
While previous studies have provided comprehensive evaluation of single-path solutions, the same level of detail is lacking when considering sets of paths. 
This paper emphasizes that the path characterization strongly depends on the properties under consideration. While property-based characterization is also valid for single paths, it becomes crucial to analyse multiple path sets. 
From the above consideration, this paper proposes a mathematical approach, defining a functional model that lends itself well to characterizing the path set in its general formulation. The paper shows how the functional model contextualizes specific attributes.
\end{abstract}


\section{Introduction}
\label{Sec:Introduction}

A network is a concept often used in several domains: a number of devices or users connected to transfer and exchange data. As a very general definition, in this paper \textit{a network is any collection of objects in which some pairs
of these objects are connected by links}~\cite{Easley2010}.
One major application field for the network concept is the telecommunications domain, particularly packet switching, whose purpose is to route packets among devices. Another field of application is transportation, where the purpose is to send vehicles through roads from one place to another. 

Routing is the process of choosing a path or a set of paths across one or more networks. Choosing paths implies considering several alternatives, giving a score to each of them, and eventually selecting the highest-scoring alternative. The described selection process is what existing routing algorithms do. 
Solid solutions exist for single-path routing, where a path is given a cost or length value, which is the sum of the costs of its component edges. When evaluating sets of multiple paths, scoring criteria are not equally clear. This paper proposes a mathematical model to determine descriptive and representative characteristics of path sets. The model has practical applications in telecommunication systems.

Section~\ref{sec:RelatedWorks} analyses how research has tackled the problem of characterizing network connectivity.
Section~\ref{Sec:GraphElements} recalls essential elements of graph theory and their interpretation for this context. Section~\ref{sec:weights} shows how edge attributes should be composed to calculate equivalent attributes in arbitrary path sets.
Section~\ref{sec:tools} introduces representations and transformations of path sets to help the property calculation. The full calculation proposal is presented in Section~\ref{sec:Math} and illustrated with notable examples. Section~\ref{Sec:polymatroids_of_vertex_weighted_hypergraphs} analyses the correlation between the paper's proposal and the network polymatroid representation, with an analysis on the discussed attributes.
Eventually, Section~\ref{sec:Conclusion} contains the paper's conclusions and possible future developments.
\section{Related Works}
\label{sec:RelatedWorks}

In telecommunications, and more generally in graph theory, analysing the best route is still a research topic. The simplest case, or single-path analysis, is a well-defined problem: finding the shortest route between two points. The “cost” could represent several properties: distance, time, money, and others. Here, algorithms such as Dijkstra~\cite{Dijkstra1959} and Bellman-Ford~\cite{Bellman1958, Ford1956} have stood the test of time. They focus on calculating the most efficient route, minimizing cost while crossing from one point to another. While most practical approaches focus on additive costs, there exist proposals considering single-path routing generalizations to arbitrary costs~\cite{Sobrinho2005}.

When searching for multiple paths, minimizing the cost alone is not sufficient. Since a single path already secures connectivity between the two points, a cost-based approach would just avoid additional paths. The reason for searching multiple paths is to improve other properties, such as fault resiliency. Previous approaches addressed the problem by adding a diversity score to the set of paths they aimed to find, sometimes implicitly. The diversity concept has several issues. Firstly, there are multiple diversity definitions between two paths. Secondly, extending the diversity concept from a pair of paths to a set of more than two is non-trivial and always risks missing some important properties of the set. Finally, all the definitions of diversity originate from another property, most often the cost, which does not necessarily relate to the property that needs improvement by adding multiple paths; in general, such an approach may mislead the search.  

Algorithms solving the K shortest path problem, such as Yen~\cite{Yen1971} or Eppstein~\cite{Eppstein1995} search for multiple paths while minimizing them with a minimum total cost, where the definition of diversity relies implicitly on the presence of at least one different edge, which in most cases is not effective. Others, such as Suurballe~\cite{Suurballe1974} and Bhandari~\cite{Bhandari1999}, focus on finding completely disjoint paths that do not share common features. 

Meanwhile, researchers such as Liu~\cite{Liu2018} and Chondrogiannis~\cite{Chondrogiannis2015} have explored how to integrate diversity more directly, defining the diversity of a path set as the minimum diversity among any two paths in the set.

Borodin et al.~\cite{Borodin2017} use diversity for high-quality result diversification in information retrieval, database management, and structure localization. 
Fiaschi et al.~\cite{Fiaschi2020} apply Borodin's concepts to path search by combining cost and diversity with a parameter tuning their relative importance.

Network graph analysis frequently utilizes the polymatroid mathematical model because it effectively captures the complex constraints of connectivity and interference while simplifying the development of optimization algorithms~\cite{Harks2014}. Notable examples where polymatroids serve as the primary mathematical model for analysing information flow are the works of Riemensberger et al.~\cite{Riemensberger2014}, Chekuri et al.~\cite{Chekuri2015}, and Lee et al.~\cite{Latham2006}.

\section{Routing Elements}
\label{Sec:GraphElements}

\begin{figure}[t!]
    \centering
    \includegraphics[width=0.8\linewidth, keepaspectratio]{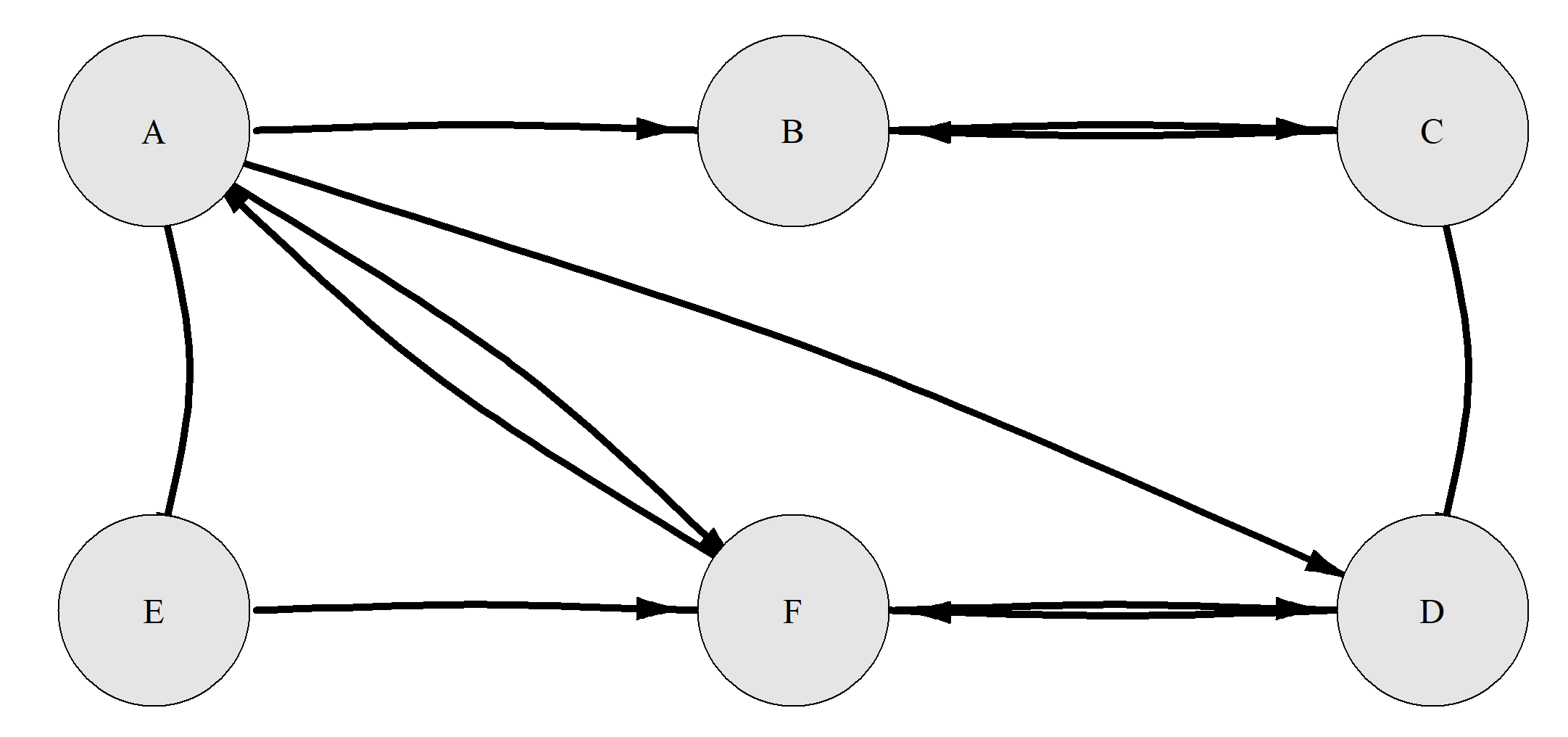}
    \caption{Generic Graph example}
    \label{fig:GenericGraphsexample}
\end{figure}

\subsection*{Graphs}
Both undirected graphs and directed graphs are considered in graph theory~\cite{Sandryhaila2014}. This paper will focus solely on directed graphs, as they are more applicable to telecommunication and transportation examples~\cite{monteiro2012}, where network links are traversed in a specific direction. 
For an abstract mathematical representation, a network is represented as a directed graph $\mathcal{G} = (V_g, E_g)$ (see Figure~\ref{fig:GenericGraphsexample}), defined by a set $V_g$ of vertices (the network nodes) and a set $E_g$ of edges (the network links), where $E_g$ is a subset of $V_g \times V_g$, the set of ordered pairs of vertices. 

\subsection*{Directed Paths}
\label{Sec:Paths}
Intuitively, a directed path in a directed graph is an essential concept in applications such as telecommunications and transportation, as it models a route from one network node to another and provides means to convey useful entities between the end nodes, such as information in telecommunications and physical vehicles in transportation.
Formally, a directed path from a source vertex $s$ to a destination vertex $d$, is a sequence of edges 
$p =$ $\{(v_0, v_1)$, ... $(v_{i-1}, v_{i})$, $(v_{i}, v_{i+1})$, ... $(v_{n-1}, v_n)\}$
such that $v_0, v_1, \ldots, v_n$ are distinct vertices, $s=v_0$, $d=v_n$ and $(v_{i-1}, v_i) \in E_g$ for $i = 1,2, \dots,n$. Other equivalent definitions may be found in literature~\cite{Latham2006, Hu2019, POTTER2007}.

\subsection*{Connections}
\label{sec:connections}
In graph theory, a vertex $x$ is connected to a vertex $y$ if a directed path exists from $x$ to $y$~\cite{Stallings2007, Kurose2021}. 
A vertex may be connected to another vertex in several different ways within a graph. For example, there might be multiple paths between the two vertices. 

\begin{Definition}
A \emph{path set} is a set of paths between the same two vertices.
\end{Definition}
\begin{Definition}
A \emph{connection} is the relationship between two vertices together with the attributes characterizing the relationship.
\end{Definition}
In other words, the definition of a connection refers to the source and destination vertices and its attributes, abstracting from the details of the paths that implement it.

\section{Routing Attributes} 
\label{sec:weights}

\subsection*{Attributes of a Graph Edge}
Edge attributes complete a graph definition, describing some property of the edge. This will help to describe more formally the properties of the connections introduced above. The weight of graph edges can be formally described as a function $w: E_g \rightarrow \mathbb{R}$ (a value of a network link). The meaning of the weight depends on its purpose in specific applications and may represent different concepts.
Examples can be delay, administrative cost, capacity, and fault probability.

\subsection*{Attributes of a Directed Path}
The weights assigned to the graph edges induce a weight on each directed path.
The weight of a path is generally defined as the sum of the weights of its edges.
However, for some edge properties it may be convenient to revise the definition according to how the property combines when crossing two or more edges in series. 

Different weight meanings require different serial operations. In the following the paper indicates the serial composition of weights with the \textit{$OP_s$} operation.

A simple example where the composition is not the sum is the capacity, where the path capacity is limited to the narrowest edge, so that \textit{$OP_s$} can be instantiated with the minimum.

\subsection*{Attributes of a Connection Implemented by Disjoint Paths}
Setting up one path between two points is essential to allow connections between these points, as when finding a route to travel from one site to another in the transportation application. The path choice determines a specific value of the path attributes, with a meaning on the characteristics of the path, like the time needed to travel from one site to another following that path.
If the value of the attribute reachable with a single path is not satisfactory for certain applications, one may want to establish several paths and achieve a change in that value. The composition of attributes using several paths in parallel follows of course a different operation with respect to the serial composition. 

Different weight meaning require different parallel operations. In the following, the \textit{$OP_p$} operation indicates the parallel composition of weights.

\subsection*{Attributes of a Connection Implemented by Arbitrary Paths}
In general, a connection may be implemented by a set partially overlapping paths. In that case, a straightforward combination of serial and parallel operations is inadequate for quantitatively representing the attribute in question. The analysis requires a suitable $T$ transformation of the path set. 
Different attributes may necessitate different transformations.
\section{Tools for Path Set Representation and Transformation} \label{sec:tools}

This chapter presents hypergraphs and their incidence matrices as a convenient tool for representing sets of sets of edges in a graph. A path set is a first example of a set of sets of edges in that a path is a set of edges. The chapter then introduces two transformations of a path set, the union and the cuts, that can also be represented with hypergraphs. Finally it uses vertex-weighted hypergraphs and their $r$-incidence matrices to incorporate attribute values in the hypergraph representation of path sets and their transformations. 

\begin{figure}[!t]
\centering
    \includegraphics[width=0.8\linewidth, keepaspectratio]{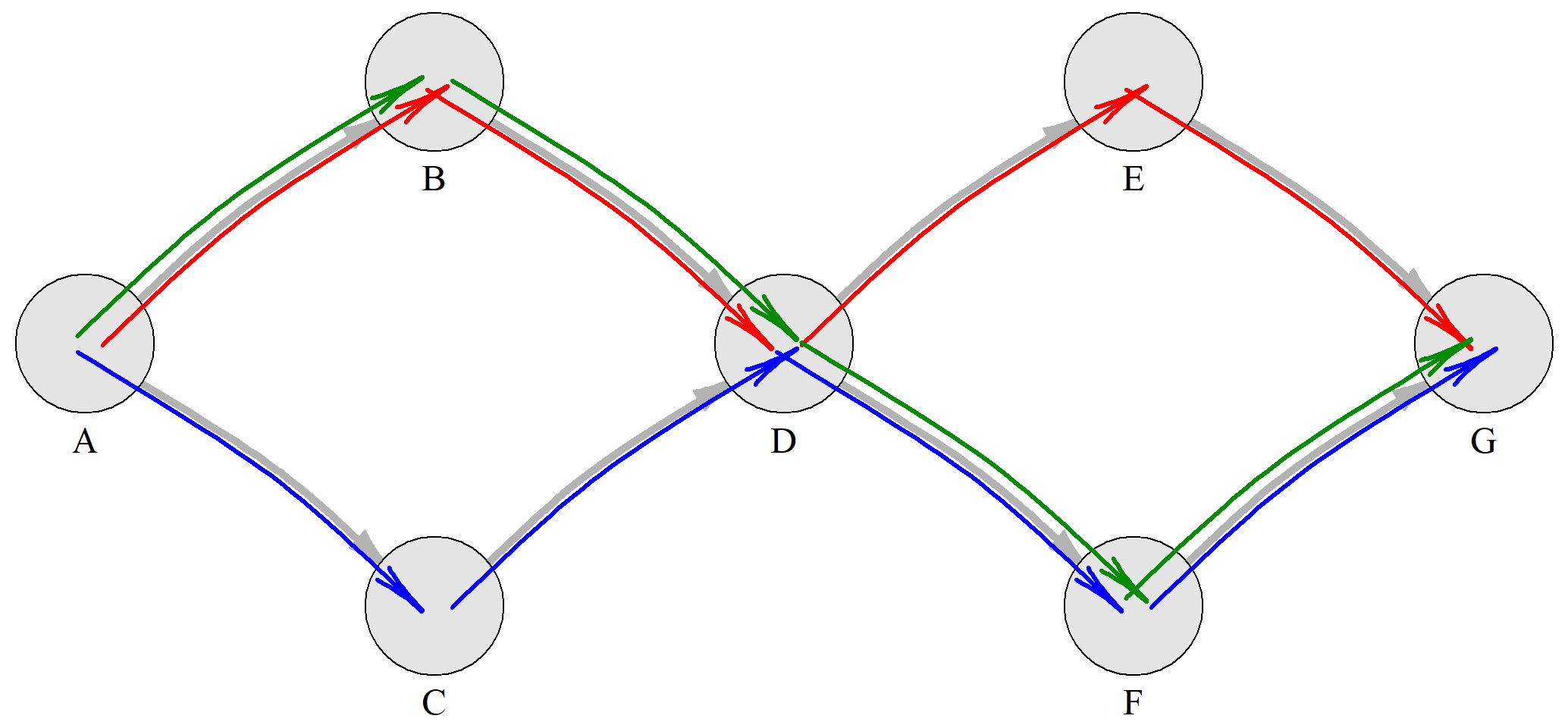}
  \caption{A graph with three paths}%
  \label{fig:3paths}
\end{figure}
%\subsection{Edge property vector}

\subsection*{Hypergraphs}

Hypergraphs were introduced by C.~Berge in the 1960s~\cite{berge1984hypergraphs} as a generalization of graphs where edges join any number of vertices.

A (finite) \textit{hypergraph} is the pair $\mathcal{H} = (V_h,E_h)$ where $V_h$ is a finite set of elements called \textit{vertices}, $E_h$ is a finite set of elements called \textit{edges}, and $e \subseteq V_h$ for each element $e \in E_h$.
Hypergraphs can be represented using incidence matrices.
The \textit{incidence matrix} $B=(b_{ij})_{m \times n}$ of a hypergraph $\mathcal{H} = (V_h,E_h)$ with $V_h=\{v_1,\ldots,v_m\}$ and $E_h=\{e_1,\ldots,e_n\}$ is defined by\\
\[
b_{ij} = \biggl\{ 
\begin{array}{ll}
1 & \hbox{if $v_i \in e_j$,}\\
0 & \hbox{otherwise.}
\end{array}
\]

A hypergraph can be convenient for describing path sets and their transformations. The next sections explain how to construct hypegraphs for specific representations.

\subsection*{Path Sets}

A set of directed paths can be described by a hypergraph, whose vertices correspond to the edges of the graph on which the paths are routed and whose edges correspond to the paths in the set.

For example, Figure \ref{fig:3paths} illustrates a set of three paths $P = \{p_1, p_2, p_3\}$ (green, red and blue in figure) defined over a graph $\mathcal{G}=(V_g, E_g)$, the three paths connecting vertices $A$ and $G$. 

In the following, the notation $\hypid{P}$ indicates the hypergraph describing path set $P$ and $\routid{P}$ its incidence matrix. In the example of Figure \ref{fig:3paths}, if $\hypid{P}=(V_h, E_h)$, then $V_h = E_g$, and $E_h = P = \{p_1, p_2, p_3\}$.

Table \ref{tab:PathSet_Transformation_tables}.B shows the incidence matrix $\routid{P}$ of $\hypid{P}$. 

\subsection*{Union}

A very simple transformation of a path set $P$ is the union of all the edges pertaining to its member paths. The union of a path set is a single set of edges, therefore it can be represented by a hypergraph with a single edge. 

The notations $\hypuni{P}$ and $\routuni{P}$ will indicate the hypergraph of the union of $P$ and its incidence matrix respectively.

Table \ref{tab:PathSet_Transformation_tables}.C shows $\routuni{P}$ for the example in Figure~\ref{fig:3paths}.

\begin{table*}[!t]
\centering
    \caption{Path set transformation tables}
\resizebox{\linewidth}{!}{%
\begin{tabular}{||l || c || c c c || c || c c c c c c c c c c c c || } 
\hline
\hline
%& Table A & & Table B & & Table C & & & & & & Table &D & & & & & \\
& A & \multicolumn{3}{c ||}{B} & C & \multicolumn{12}{c ||}{D} \\
\hline
& $\mathbf{w}$ & \multicolumn{3}{c ||}{$\routid{P}$} & $\routuni{P}$ & \multicolumn{12}{c ||}{$\routcut{P}$}\\
\hline
    Edge & Weight & P1 & P2 & P3 & Union & C1 & C2 & C3 & C4 & C5 & C6 & C7 & C8 & C9 & C10 & C11 & C12 \\
\hline
    A→B & $w_1$ & 1 & 1 & 0 & 1 & 1 &   1 &   1 &   1 &   0 &   0 &   0 &   0 &   0 &   0 &   0 &   0 \\
    A→C & $w_2$ & 0 & 0 & 1 & 1 & 1 &   0 &   0 &   0 &   1 &   0 &   0 &   0 &   0 &   0 &   0 &   0 \\
    B→D & $w_3$ & 1 & 1 & 0 & 1 & 0 &   0 &   0 &   0 &   1 &   1 &   1 &   1 &   0 &   0 &   0 &   0 \\ 
    C→D & $w_4$ & 0 & 0 & 1 & 1 & 0 &   1 &   0 &   0 &   0 &   1 &   0 &   0 &   0 &   0 &   0 &   0 \\
    D→E & $w_5$ & 1 & 0 & 0 & 1 & 0 &   0 &   0 &   0 &   0 &   0 &   0 &   0 &   1 &   1 &   0 &   0 \\
    D→F & $w_6$ & 0 & 1 & 1 & 1 & 0 &   0 &   1 &   0 &   0 &   0 &   1 &   0 &   1 &   0 &   1 &   0 \\
    E→G & $w_7$ & 1 & 0 & 0 & 1 & 0 &   0 &   0 &   0 &   0 &   0 &   0 &   0 &   0 &   0 &   1 &   1 \\
    F→G & $w_8$ & 0 & 1 & 1 & 1 & 0 &   0 &   0 &   1 &   0 &   0 &   0 &   1 &   0 &   1 &   0 &   1 \\ 
    \hline
    \hline
    \end{tabular}%
    }
    \label{tab:PathSet_Transformation_tables}
\end{table*}

\begin{table*}[!b]
\centering
    \caption{\textit{r}-incidence matrices}
\setlength{\tabcolsep}{2pt}    
\resizebox{\linewidth}{!}{%
\begin{tabular}{|| l || c c || c c c c c c c c c c c c || c c c c c c c c c c c c || } 
\hline
\hline
& \multicolumn{2}{c ||}{A} & \multicolumn{12}{c ||}{B} & \multicolumn{12}{c ||}{C} \\
\hline
& Cap. & Prob. & \multicolumn{12}{c ||}{$R(0, c, \routcut{P})$} & \multicolumn{12}{c ||}{$R(1, p, \routcut{P})$} \\
\hline
    Edge & $\mathbf{c}$ & $\mathbf{p}$ & C1 & C2 & C3 & C4 & C5 & C6 & C7 & C8 & C9 & C10 & C11 & C12 & C1 & C2 & C3 & C4 & C5 & C6 & C7 & C8 & C9 & C10 & C11 & C12 \\
\hline
    A→B & 
    25 & 0.0050 &
    25 & 25 & 25 & 25 & 0 & 0 & 0 & 0 & 0 & 0 & 0 & 0 & 
    0.0050 & 0.0050 & 0.0050 & 0.0050 & 1 & 1 & 1 & 1 & 1 & 1 & 1 & 1 \\
    A→C & 
    10 & 0.0075 &
    10 & 0 & 0 & 0 & 10 & 0 & 0 & 0 & 0 & 0 & 0 & 0 & 
    0.0075 & 1 & 1 & 1 & 0.0075 & 1 & 1 & 1 & 1 & 1 & 1 & 1 \\
    B→D & 
    25 & 0.0070 &
    0 & 0 & 0 & 0 & 25 & 25 & 25 & 25 & 0 & 0 & 0 & 0 & 
    1 & 1 & 1 & 1 & 0.0070 & 0.0070 & 0.0070 & 0.0070 & 1 & 1 & 1 & 1 \\ 
    C→D & 
    25 & 0.0040 &
    0 & 25 & 0 & 0 & 0 & 25 & 0 & 0 & 0 & 0 & 0 & 0 & 
    1 & 0.0040 & 1 & 1 & 1 & 0.0040 & 1 & 1 & 1 & 1 & 1 & 1 \\
    D→E & 
    10 & 0.0115 &
    0 & 0 & 0 & 0 & 0 & 0 & 0 & 0 & 10 & 10 & 0 & 0 & 
    1 & 1 & 1 & 1 & 1 & 1 & 1 & 1 & 0.0115 & 0.0115 & 1 & 1 \\
    D→F & 
    25 & 0.0045 &
    0 & 0 & 25 & 0 & 0 & 0 & 25 & 0 & 25 & 0 & 25 & 0 & 
    1 & 1 & 0.0045 & 1 & 1 & 1 & 0.0045 & 1 & 0.0045 & 1 & 0.0045 & 1 \\
    E→G & 
    10 & 0.0105 &
    0 & 0 & 0 & 0 & 0 & 0 & 0 & 0 & 0 & 0 & 10 & 10 & 
    1 & 1 & 1 & 1 & 1 & 1 & 1 & 1 & 1 & 1 & 0.0105 & 0.0105 \\
    F→G & 
    100 & 0.0065 &
    0 & 0 & 0 & 100 & 0 & 0 & 0 & 100 & 0 & 100 & 0 & 100 & 
    1 & 1 & 1 & 0.0065 & 1 & 1 & 1 & 0.0065 & 1 & 0.0065 & 1 & 0.0065 \\ 
    \hline
    \hline
    \end{tabular}%
    }
    \label{tab:r-incidence}
\end{table*}

\subsection*{Cuts}

A cut is a well-known concept in graph theory: given a partition of the vertices into two disjoint subsets, a graph cut is the set of all edges that have one end vertex in each partition of the cut~\cite{Diestel2000}. 

Given a path set connecting two vertices $S$ and $D$, this paper introduces the expression \textit{cut of the path set} for a minimal edge set that, if removed, will make $S$ and $D$ unreachable within the path set. The cut is minimal in the sense that all the edges in the cut must be removed from the path set to make the vertices unreachable. 

For example, in Figure~\ref{fig:3paths}, the set of edges $B \rightarrow D$ and $F \rightarrow G$ is a cut of the path set, but is not a graph cut, because it is still possible to reach $G$ from $A$. 

The notations $\hypcut{P}$ and $\routcut{P}$ will indicate the hypergraph of the cuts of $P$ and its incidence matrix respectively.

Table \ref{tab:PathSet_Transformation_tables}.D shows $\routcut{P}$ for the example in Figure~\ref{fig:3paths}.

\subsection*{Edge Property Vector}
A column vector $\mathbf{w}$ encodes the edge properties by associating each edge with a value corresponding to the property of interest.
This vector will come handy to distribute the property values over matrices representing a path set or one of its transformations. 
Table \ref{tab:PathSet_Transformation_tables}.A shows a vector $\mathbf{w} = (w_1,\ldots,w_8)$ of properties for the graph in Figure~\ref{fig:3paths}.

\subsection*{Vertex-Weighted Hypergraphs}
This section defines vertex-weighted hypergraphs, which are a generalization of hypergraphs. For vertex-weighted hypergraphs, the paper also introduces \( r \)-incidence matrices, which generalize the concept of incidence matrices for hypergraphs.

A \textit{vertex-weighted hypergraph} is a tuple $\mathcal{H} = (V_h,E_h,w)$ where $(V_h,E_h)$ is a hypergraph and $w:V_h \rightarrow \mathbb{R}$ is a weight function which assign a real number $w(v)$ to each vertex $v \in V_h$.

If a hypergraph represents a set of sets of edges in a graph $\mathcal{G}$, and $\mathcal{G}$ has an associated edge property vector $\mathbf{w}$, then a vertex-weighted hypergraph can be easily constructed by associating $\mathbf{w}$ to the vertices of the hypergraph, since they correspond to the edges of $\mathcal{G}$.
\hfill\break

For any real number $r \in \mathbb{R}$, let the \textit{r-incidence matrix} $B_r = (b_{ij})_{m \times n}$ of a vertex-weighted hypergraph $\mathcal{H} = (V_h,E_h,w)$ with $V_h = \{v_1, \ldots, v_m\}$ and $E_h = \{e_1, \ldots, e_n\}$ be defined by\\
\[
b_{ij} = \biggl\{ 
\begin{array}{ll}
w(v_i) & \hbox{if $v_i \in e_j$,}\\
r & \hbox{otherwise.}
\end{array}
\]

Note that the 0-incidence matrix $B_0$ of a vertex-weighted hypergraph $\mathcal{H} = (V_h,E_h,w)$ with $w(v) = 1$ for all $v \in V$ equals the common incidence matrix $B$ of the hypergraph $(V_h,E_h)$.

The $r$-incidence matrices of special interest for this paper are the 0-incidence and 1-incidence matrices.

The notation $B_r = R(r, w, B)$ denotes the \(r\)-incidence matrix of a vertex-weighted hypergraph $\mathcal{H} = (V_h, E_h, w)$, where $B$ represents the incidence matrix of the underlying hypergraph $(V_h, E_h)$ without vertex weights.

Table~\ref{tab:r-incidence} shows \(r\)-incidence matrix examples for hypergraphs representing cut transformation of the path set in Figure~\ref{fig:3paths} with edge property vectors $\mathbf{c}$ (capacity) and $\mathbf{p}$ (unavailability or fault probability), as shown in Table~\ref{tab:r-incidence}.A.

\section{Calculation of Properties}
\label{sec:Math}

The conclusion of Section~\ref{sec:weights} suggests that, in a specific domain, a mathematical model characterizes a path set independently of the attribute under analysis. The mathematical model represents the intrinsic contributions of the path set: 
\begin{itemize}
    \item an operation that represents the contribution to the characterization due to the serialization of the edges in a path,
    \item an operation that represents the contribution to the characterization due to the parallelization of the paths in a set of disjoint paths,
    \item a transformation combing the two operations in an arbitrary path set.
\end{itemize}

Equation~\ref{equ:MathFramework} summarizes these three essential elements of the mathematical model, although not necessarily in the order laid down:
\vspace{-2mm}
\begin{equation}
\label{equ:MathFramework}
    \Psi(P) = OP_s \circ OP_p \circ T
\end{equation}

With the tools presented so far, it is now possible to present the complete attribute calculation in specific examples with serial and parallel compositions together with transformations.
\hfill\break

Table~\ref{tab:transformationsummary} shows a summary of five notable examples.

\begin{table}[!t]
\caption{Operations and transformations}
\label{tab:transformationsummary}
\centering
\small
\resizebox{\columnwidth}{!}{%
\begin{tabular}{||l|ccc||} 
 \hline
 \hline
 Example & $OP_s$ & $OP_p$ & T \\ [0.5ex] 
 \hline
 \hline
 Delay & $+$ & $max$ & $R(0, w, \routid{P})$ \\
\hline
 Administrative& $+$ & $+$ & $R(0, w, \routuni{P})$ \\
 cost&  &  &  \\
 \hline
 Combined& $min$ & $+$ & $R(0, w, \routcut{P})$ \\
 capacity&  &  &  \\
 \hline
 Unavailability & $1-((1-(.))\times$ & $\times$ & $R(1, w, \routcut{P})$  \\
  & $\times(1-(.)))$ &  &  \\
 \hline
 Fault& $+$ & $\times$ & $R(1, w, \routcut{P})$ \\ 
 probability&  &  &  \\ 
 \hline
 \hline
\end{tabular}
}
\end{table}

Table~\ref{tab:exampleinstatiations} shows their instantiations in a more direct notation. 

\begin{table}[!b]
\caption{Property examples instantiations}
\label{tab:exampleinstatiations}
\centering
\resizebox{\linewidth}{!}{
\begin{tabular}{||l | l ||}
\hline\hline
Delay($P$) & $\max_{i | p_i \in P}{\sum_{j | e_j \in p_i}{w_j}}\text{,}$ \\[1em]
Cost($P$) & $\sum_{i | e_i \in \bigcup P}{w_i}\text{,}$ \\[1em]
Capacity($P$) & $\min_{j}{\sum_{i}{T_{ij}}}\text{, where } \mathbf{T} = R(0, w, \routcut{P})\text{,}$ \\[1em]
UnavailProb($P$) & $1 - \prod_{j}{\left(1 - \prod_{i}{T_{ij}}\right)}\text{, where } \mathbf{T} = R(1, w, \routcut{P})$, \\[1em]
FaultProb($P$) & $\sum_{j}{\prod_{i}{T_{ij}}}\text{, where } \mathbf{T} = R(1, w, \routcut{P})$. \\
\hline\hline
\end{tabular}
}
\end{table}

The serial operations in the given examples are:
\begin{itemize}
    \item Sum for the delay (edge delays are simply summed)
    \item Sum for the administrative cost
    \item Minimum for the capacity (bottleneck edge)
    \item One's complement to product of the one's complements for the unavailability (because of probability theory~\cite{Dubrova2013})
    \item Sum for the fault probability (as simplification of the unavailability neglecting the higher order components~\cite{Dubrova2013})
\end{itemize}

The parallel operations in the given examples are:
\begin{itemize}
    \item Max for the delay (assuming to wait for the latest packet)
    \item Sum for the administrative cost
    \item Sum for the capacity
    \item Product for the unavailability and fault probability (because of probability theory~\cite{Dubrova2013})
\end{itemize}

The path set transformations are:
\begin{itemize}
    \item Identity for the delay, then serial first
    \item Union for the administrative cost
    \item Cuts for the capacity, unavailability and fault probability, then parallel first
\end{itemize}

The $r$ value in the \textit{r}-incidence matrices shall be the identity element of the operation to be performed first, therefore:
\begin{itemize}
    \item 0 for delay, administrative cost and capacity (sum first)
    \item 1 for unavailability and fault probability (product first)
\end{itemize}

Five examples of path set properties illustrate the application of the method: delay, administrative cost, capacity, unavailability, and fault probability. Table~\ref{tab:propvectexamples} shows their edge property vectors.

\subsection*{Delay}

Equation~\ref{Equ:example_transformation_delay}
shows the mathematical application example for $\mathbf{T}$ matrix in the delay calculation. 
\vspace{-2mm}
\[
\scalebox{0.9}{$
    Delay(P) = \max{\sum{_i}{T_{ij}}} = \ 
    \max{([340, 230, 225])} = \ 340\mu s
$}
\]
\vspace{-3mm}
\begin{equation} 
\label{Equ:example_transformation_delay}
\resizebox{0.7\columnwidth}{!}{$
\begin{aligned}
    \mathbf{T} = & R(0, w, \routid{P}) = \\ 
    & R \left(
    0,
    \begin{bmatrix}
    50 \\ 75 \\ 70 \\ 40 \\ 115 \\ 45 \\ 105 \\ 65
    \end{bmatrix}
    ,
    \begin{bmatrix}
     1 &   1 &   0 \\ 
     0 &   0 &   1 \\ 
     1 &   1 &   0 \\ 
     0 &   0 &   1 \\ 
     1 &   0 &   0 \\ 
     0 &   1 &   1 \\ 
     1 &   0 &   0 \\ 
     0 &   1 &   1 \\ 
    \end{bmatrix} 
    \right) =
%    \hspace{3mm}
    \begin{bmatrix}
     50 &   50 &   0 \\ 
     0 &   0 &   75 \\ 
     70 &   70 &   0 \\ 
     0 &   0 &   40 \\ 
     115 &   0 &   0 \\ 
     0 &   45 &   45 \\ 
     105 &   0 &   0 \\ 
     0 &  65 &  65 \\ 
    \end{bmatrix}
\end{aligned}
$}
\end{equation}

\begin{table}[!b]
\caption{Example edge property vectors}
\label{tab:propvectexamples}
\small
\centering
\resizebox{0.9\linewidth}{!}{
\begin{tabular}{||c|cccc||}
 \hline
 \hline
 & Delay
 & Admin. 
 & Capacity
 & Fault \\ 
 & (µs)
 & cost 
 & (Gbps)
 & probability \\ 
  \hline
  A→B & 50 & 100 &  25 & 0.0050 \\ 
  A→C & 75 & 150 &  10 & 0.0075 \\ 
  B→D & 70 & 300 &  25 & 0.0070 \\ 
  C→D & 40 & 200 &  25 & 0.0040 \\ 
  D→E & 115 & 500 &  10 & 0.0115 \\ 
  D→F & 45 & 100 &  25 & 0.0045 \\ 
  E→G & 105 & 450 &  10 & 0.0105 \\ 
  F→G & 65 & 100 &  100 & 0.0065\\ 
   \hline
 \hline
\end{tabular}
}
\end{table}

\subsection*{Administrative Cost}

Equation~\ref{Equ:example_transformation_cost}
shows the numerical value of the administrative cost. $\mathbf{T} =  R(0, w, \routuni{P})$, where Table~\ref{tab:PathSet_Transformation_tables}.C shows $\routuni{P}$.

\vspace{-1mm}
\begin{equation}
\label{Equ:example_transformation_cost}
\resizebox{0.9\columnwidth}{!}{$
\begin{aligned}
    Cost(P) = & \sum{}\mathbf{T} =\\
    & \sum{}{ \left(
    \begin{bmatrix}
    100 , 150 , 300 , 200 , 500 , 100 , 450 , 100
    \end{bmatrix} \right)
    } = \\
    & 1900
\end{aligned}
$}
\end{equation}

\subsection*{Capacity}

Equation~\ref{Equ:example_transformation_capacity}
shows the numerical value of the capacity. $\mathbf{T} = \ R(0, w, \routcut{P})$, as described in Table~\ref{tab:r-incidence}.B. 

\begin{equation}
\label{Equ:example_transformation_capacity}
\resizebox{0.7\columnwidth}{!}{$
\begin{aligned}
     Capacity(P) =   
     & \min\sum{_i}{T_{ij}} = \\ 
     & \min([35, 50, 50, 125, 35, 50,\\ 
     & 50, 125, 35, 110, 35, 110]) = \\ 
     & 35 Gbps
\end{aligned}
$}
\end{equation}

\subsection*{Unavailability and Fault Probability}

Equation~\ref{Equ:example_transformation_prob}
shows the numerical values of unavailability and fault probability. $\mathbf{T} = \ R(1, w, \routcut{P})$, as described in Table~\ref{tab:r-incidence}.C.
\vspace{-2mm}

\begin{equation}
\label{Equ:example_transformation_prob}
\resizebox{0.7\columnwidth}{!}{$
\begin{aligned}
    UnavailProb(P) = & 1 - \prod_{j}{\left(1 - \prod_{i}{T_{ij}}\right)} = \\ 
    & 0.0005118815
    &\\
    &\\
    FaultProb(P) = & \sum{_j}\prod{_i}{T_{ij}} = \\
    & 0.000512
\end{aligned}
$}
\end{equation}

where $\prod_{i}{T_{ij}}$ is:
\begin{equation*}
\resizebox{0.6\columnwidth}{!}{$
\begin{aligned}
    [ &
    0.00003750 , 
    0.00002000 , 
    0.00002250 ,\\ & 
    0.00003250 , 
    0.00005250 , 
    0.00002800 ,\\ & 
    0.00003150 , 
    0.00004550 , 
    0.00005175 ,\\ & 
    0.00007475 , 
    0.00004725 , 
    0.00006825 
    ] 
\end{aligned}
$}
\end{equation*}

It is worth noting how the fault probability calculation perfectly matches the HAM method described by Vitucci et al.~\cite{Vitucci2024} using a more "empirical" method.
\section{Polymatroids}
\label{Sec:polymatroids_of_vertex_weighted_hypergraphs}

Matroids, a field in algebraic combinatorics, and the generalization of matroids, polymatroids, are research areas with applications and connections to many other areas in mathematics and computer science. For example, polymatroid theory has proven to be very useful in combinatorial optimization~\cite{Schrijver2002-fy}. One of the strengths with matroids and polymatroids is the many equivalent axiom systems from different areas of mathematics that can be used to define these objects. This paper defines matroids and polymatroids using submodular set functions.

\begin{Definition}
\label{def:Polymatroid}
    Let $E$ be a finite set. A pair $\mathcal{P} = (\rho, E)$ is a (finite) \emph{polymatroid} on $E$ with set function $\rho: 2^E \rightarrow \mathbb{R}$ if $\rho$ satisfies the three conditions (R1), (R2), and (R3) for all $X,Y \subseteq E$. Note that a \emph{matroid} is a polymatroid which additionally satisfies the two conditions (R4) and (R5) for all $X \subseteq E$.
\[
\begin{array}{clcl}
(R1) & \rho(\emptyset) = 0, & (R4) & \rho(X) \in \mathbb{Z},  \\
(R2) & X \subseteq Y \Rightarrow \rho(X) \leq \rho(Y), & (R5) & \rho(X) \leq |X|, \\
(R3) & \rho(X) + \rho(Y) \geq \rho(X \cap Y) + & \\
& \rho(X \cup Y).& &
\end{array}
\]
\end{Definition}

\subsection*{Submodular and Supermodular Considerations}
Polymatroids have proven to be a highly useful mathematical tool in various areas of mathematics and computer science. 
In this chapter, let $E$ be a set of directed paths from a vertex $s$ to a vertex $d$ in a graph associated with a network.
Analysing delay, cost, capacity, unavailability probability, and fault probability of a path set $P \subseteq E$ leads to the definition of set functions $\rho: 2^E \rightarrow \mathbb{R}$.
The analysis of these set functions shows that the pair $(\rho, E)$ forms a polymatroid with respect to delay and cost, satisfying conditions (R1)–(R3) in Definition \ref{def:Polymatroid}.
However, for the concepts of capacity, unavailability probability, and fault probability, the situation is different.
A set function $\rho: 2^E \rightarrow \mathbb{R}$ is said to be \emph{submodular} if (R3) in Definition \ref{def:Polymatroid} holds, i.e., 
$
\rho(X) + \rho(Y) \geq \rho(X \cap Y) + \rho(X \cup Y)
$
for all $X, Y \subseteq E$, and $\rho$ is said to be \emph{supermodular} when it satisfies the opposite condition to submodular, i.e., 
$
\rho(X) + \rho(Y) \leq \rho(X \cap Y) + \rho(X \cup Y)
$
for all $X, Y \subseteq E$. A supermodular set function $\rho$ satisfying (R1) and (R2) can be transformed into a submodular set function $\rho': 2^E \rightarrow \mathbb{R}$ that satisfies (R1) and (R2) by defining, for all $X \subseteq E$,  
\begin{equation} 
\label{eq:supermodular_to_submodular}
\rho'(X) = \rho(E)-\rho(E-X)
\end{equation}
This transformation ensures that $(\rho', E)$ forms a polymatroid.

This paper examines the concepts of delay, cost, capacity, unavailability, and fault probability from a polymatrodial point of view. Interest in these mathematical structures stems from the capacity of polymatroid theory to provide a valuable framework for future research on path set attribute calculations in network systems.

\begin{Definition}
Let $E$ be a set of directed paths from a vertex $s$ to a vertex $d$ in the graph associated with a network. Furthermore, let $\weightdel$, $\weightcost$, $\weightcap$, $\weightupr$, and $\weightfpr$ be the weight functions on the edges of the graph associated with the concepts delay, cost, capacity, unavailability probability, and fault probability, respectively, on the network. 
For any subset $P \subseteq E$, let
\[
\begin{array}{clcl}
(i) & \rhodel(P) &=& \text{Delay}(P),\\
(ii) & \rhocost(P) &=& \text{Cost}(P),\\
(iii) & \rhocap(P) &=& \text{Capacity}(P),\\
(iv) & \rhoupr(P) &=& \text{UnavailProb}(P),\\
(v) & \rhofpr(P) &=& \text{FaultProb}(P),
\end{array}
\]
where $\text{Delay}(P)$, $\text{Cost}(P)$, $\text{Capacity}(P)$, $\text{UnavailProb}(P)$, and $\text{FaultProb}(P)$ are defined as in Table~\ref{tab:exampleinstatiations} with weight functions $\weightdel$, $\weightcost$, $\weightcap$, $\weightupr$, and $\weightfpr$, respectively.
\end{Definition}

\subsection*{Delay and Cost}

\begin{Theorem}
Let $E$ be a set of directed paths from a vertex $s$ to a vertex $d$ in the graph associated with a network. Then
\begin{tabular}{cl}
(i) & $(\rhodel, E)$ is a polymatroid,\\
(ii) & $(\rhocost, E)$ is a polymatroid.
\end{tabular}
\end{Theorem}

\begin{proof}
Proving that a set function constitutes a polymatroid requires demonstrating that the set function satisfies conditions (R1)–(R3) in Definition \ref{def:Polymatroid}.  For any path set $A \subseteq E$, let $\mathcal{E}(A)$ denote the set of edges in the path set $A$. 

Let $A, B \subseteq E$.

Statement (i):

(R1): Since the sum of an empty sum equals zero, it follows immediately that $\rhodel(\emptyset) = 0$.

(R2): Assume that $A \subseteq B$. Then
\[
\rhodel(A) = \max\limits_{a \in A} \rhodel(a) \leq \max\limits_{b \in B} \rhodel(b) = \rhodel(B).
\]

(R3): First, observe that
\[
\begin{array}{lcl}
\rhodel(A \cup B) &=& \max\limits_{e \in A \cup B} \rhodel(e) \\
&=& \max(\max\limits_{a \in A} \rhodel(a), \max\limits_{b \in B} \rhodel(b)) \\
&=& \max(\rhodel(A), \rhodel(B)).
\end{array}
\]
Without loss of generality, assume that $\rhodel(A \cup B) = \rhodel(A)$. Then
\[
\begin{array}{lcl}
\rhodel(A) + \rhodel(B) &=& \rhodel(A \cup B) + \rhodel(B) \\
& \overset{(R2)}{\geq} & \rhodel(A \cup B) + \rhodel(A \cap B).
\end{array}
\]

Statement (ii): Observe that
\[
\rhocost(A) = \sum\limits_{x \in \mathcal{E}(A)} \weightcost(x).
\]

(R1): Since the sum of an empty sum equals zero, it follows immediately that $\rhocost(\emptyset) = 0$.

(R2): Assume that $A \subseteq B$. Then
\[
\rhocost(A) = \sum\limits_{x \in \mathcal{E}(A)} \weightcost(x)\leq 
\sum\limits_{y \in \mathcal{E}(B)} \weightcost(y) = \rhocost(B).
\]
\hfill\break
(R3): The following shows that $\rhocost$ is both submodular and supermodular, i.e., modular,
\[
\begin{array}{lc}
\rhocost(A) + \rhocost(B) &=\\
\sum\limits_{x \in \mathcal{E}(A)}\weightcost(x) + 
\sum\limits_{y \in \mathcal{E}(B)}\weightcost(y) & =\\
\sum\limits_{x \in \mathcal{E}
(A \cup B)}\weightcost(x) + 
\sum\limits_{y \in \mathcal{E}(A \cap B)}\weightcost(y) & =\\
\rhocost(A\cup B) + \rhocost(A \cap B).
\end{array}
\]
\end{proof}

\subsection*{Capacity}

The following example shows that $\rhocap$ is, in general, neither submodular nor supermodular.

\begin{Example}
In Figure~\ref{fig:ThomasCap}, let $p_1$, $p_2$, and $p_3$ represent the red, green, and blue paths from node $a$ to node $e$, respectively, and let the capacities of all the edges be 1.

\begin{figure}
    \centering
    \includegraphics[width=\linewidth]{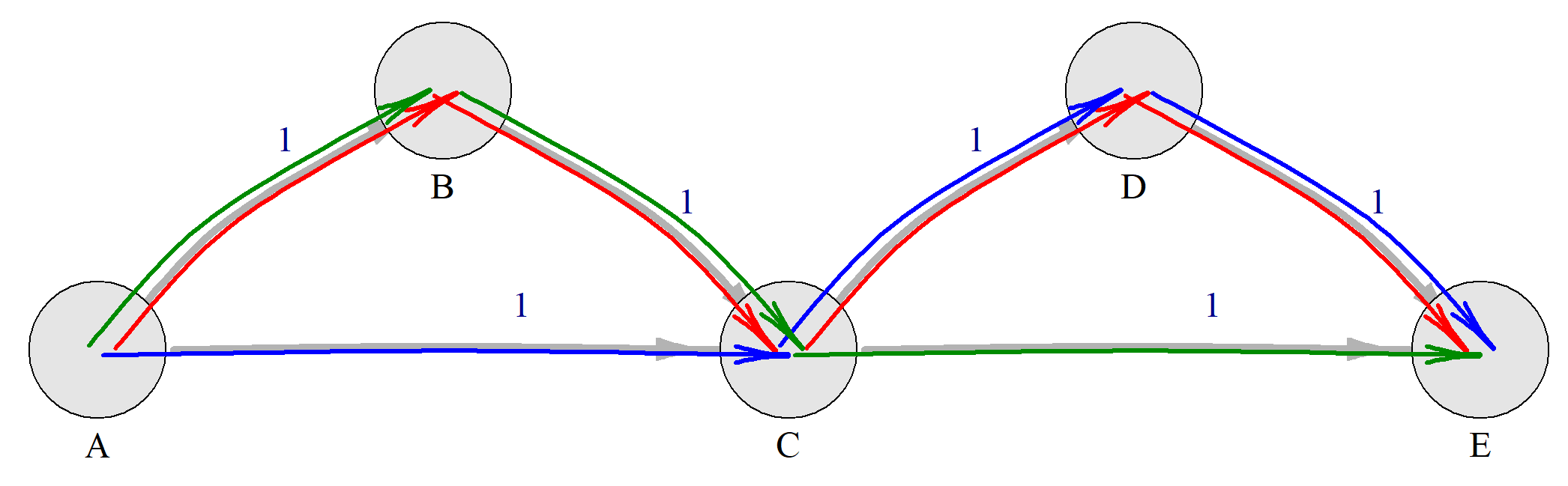}
    \caption{Capacity example}
    \label{fig:ThomasCap}
\end{figure}

If $A$ is the path set $\{p_1, p_2\}$ and $B$ is the path set $\{p_1, p_3\}$, then
\[
\begin{array}{lcl}
\rhocap(A) + \rhocap(B) &=& 1 + 1 \\
&<& 2 + 1 \\
&=& \rhocap(A \cup B) + \rhocap(A \cap B).
\end{array}
\]
\hfill\break

If instead $A$ is the path set $\{p_1\}$ and $B$ is the path set $\{p_2\}$, then
\[
\begin{array}{lcl}
\rhocap(A) + \rhocap(B) &=& 1 + 1 \\
&>& 1 + 0 \\
&=& \rhocap(A \cup B) + \rhocap(A \cap B).
\end{array}
\]
\end{Example}
\hfill\break
\subsection*{Unavailability and Fault Probability}
Both the unavailability $\rhoupr(P)$ and the Fault probability $\rhofpr(P)$ do not satisfy the polymatroid property (R1) of Definition~\ref{def:Polymatroid}. If a set is empty, there is no connection available between two vertices, with the natural consequence that $\rhoupr(\emptyset) = 1$. Following the same principle, if a set is empty there is no working path between two vertices, so even $\rhofpr(\emptyset) = 1$.

However, for the probability of the complementary events, the Availability function $\rhoapr(P)$ and the Serviceability function $\rhospr(P)$, would not suffer by the same issue. %It would then be interesting to analyse whether Availability and Serviceability are non-decreasing and submodular.
Investigating about non-decreasing property and submodularity of Availability and Serviceability requires deeper analysis and is considered for further study.
\hfill\break

As a summary:
\begin{itemize}
\item The associated functions to the concept of delay and cost have the polymatroid properties.
\item The associated function to the capacity is neither submodular nor supermodular. This means that the function is not polymatroidial and cannot, in general, be transformed into a polymatroidial set function using equation~\eqref{eq:supermodular_to_submodular}.
\item The associated functions to the unavailability and fault probability do not have the polymatroid properties. However, their transformation into a polymatroidial set function need further analysis. For example, by considering the complementary events and associated function Availability and Serviceability.
\end{itemize}

\section{Conclusions and Future Works}
\label{sec:Conclusion}

Multi-path computation in literature lacks a formal approach to quantify multiple parameters and characterizations of a path set.  
This paper proposes a paradigm shift. By addressing specific properties of graph edges, it introduces an evaluation of vertex pair connectivity that

\begin{itemize}
\item adapts to edge property composition peculiarities, generalizing to non-additive properties, and 
\item generalizes the computation to path sets and preserves the focus on the property.
\end{itemize}

The paper contextualizes the proposed solution for five different functional specifications: delay, administrative cost, capacity, unavailability, and fault probability.

The implementation of the functional model was not covered in this work. It may be the subject of a future study, that will also allow the measurement of its complexity, which was not possible to explore in this paper.

Beyond its immediate applications, this solution opens the door to automation in path selection.

By encapsulating functional characteristics in a meta-function, this model lays the groundwork for future algorithms that can dynamically identify optimal path sets, promising to enhance efficiency and significantly reduce manual workload in data analysis.

A polymatroidal analysis of the function model would support an easier evaluation of network configurations. The paper shows the polymatroidal characterization for delay and administrative cost function. Further analysis of Serviceability and Availability is another natural evolution of this paper.    
\hfill\break

\section*{Acknowledgement}
The work presented in this paper is sponsored by Ericsson, Mälardalen University and the Swedish Knowledge Foundation (KKS), via the industrial PhD School ARRAY.

\newpage

\balance
\bibliographystyle{IEEEtran}
\bibliography{IEEEabrv,Paperbiblio}
%\end{thebibliography}

\end{document}